\documentclass[10pt,letterpaper, oneside, journal, final, twocolumn]{IEEEtranTCOM}
\DeclareMathSizes{10}{9}{6}{5}

\normalsize

\usepackage{amsmath}
\usepackage{amssymb}
\usepackage{multirow}
\usepackage{hhline}
\usepackage{latexsym}
\usepackage{epsfig}
\usepackage{graphicx}
\usepackage{array}
\usepackage[usenames]{color}
\usepackage[ruled,vlined]{algorithm2e}
\usepackage{algpseudocode}
\usepackage{placeins}
\usepackage{subfigure}
\usepackage{scalerel}

\newcommand     \Eh             {{\mathbb{E}}}

\newcommand	\Vc		{{\cal V}}

\newtheorem {prop}{\bf{Proposition}}
\newtheorem {assump}{Assumption}
\newtheorem {lemma}{\bf Lemma}

\title{\LARGE 
\vspace{-5mm}
A Low-Complexity Location-Based Hybrid Multiple Access and Relay Selection in V2X Multicast Communications
\vspace{-0mm}
}
\author{
Chun-Yi Wei\thanks{
C.-Y.~Wei is with the Dept.~of Communication Engineering, National Taipei University, 23741, Taiwan (R.O.C.) (email: cywei@gm.ntpu.edu.tw).
}
, {\it  Senior Member}, {\it IEEE}, Yung Kai Wang
\vspace*{-9mm}
}
\begin{document}
\markboth{}%
{Submitted paper}
\maketitle
\begin{abstract}
This study investigated relay-assisted mode 1 sidelink (SL) multicast transmission, which encounters interference from mode 2 SL transmission, for use in low-latency vehicle-to-everything communications.  
To accommodate mode 1--mode 2 SL traffic, we use the hybrid multiple access (MA) approach, which combines orthogonal MA (OMA) and nonorthogonal MA (NOMA) schemes. We introduce a low-complexity location-based hybrid MA algorithm and its associated relay selection that can be used when SL channel state information is unavailable. 
\end{abstract}
\vspace{-2mm}
\section{Introduction}
\label{sec:intro}
Regarding 5th-generation (5G) vehicle-to-everything (V2X) communications, the forthcoming 3GPP Rel-17 standards indicate that support of the sidelink (SL) relay and of multicast transmission in the physical layer are essential \cite{Tdoc_v2x_mc,MHC_21}. 
According to 3GPP specifications \cite{MHC_21, cywei_v2x18}, the gNB [i.e., the 5G base station (BS)], may coordinate the SL communications of mode 1 user equipment (UE), also known as the centralized mode, whereas the mode 2 UE is self-originated and is known as the distributed mode. In this work, we consider a case wherein mode 1 and mode 2 SL communication coexist \cite{TRv2x_r16}. In some situations, such as public safety related use cases, the gNB requires the delivery of traffic-related messages to a group of vehicle UEs (v-UEs) within a designated proximity and in a specified time period. According to the 3GPP long-term evolution (LTE) specifications \cite{MHC_21}, the gNB may designate a relay v-UE to convey the messages using SL multicast transmission to specified group members. 
Furthermore, a properly selected relay can not only offload the heavy traffic of the Uu link but also achieve more efficient signal transmission in V2X applications, especially when the traditional Uu-link and SL-link of V2X can be allocated in different frequency bands. \cite{cs_mag21}. Unfortunately, the gNB is usually far from the site of ongoing SL communications. It has minimal knowledge regarding the SL channel state information (CSI) of all links among the anticipated v-UEs (i.e., multicast members). 
The SL-CSI feedback in SL multicast transmission is impractical, especially in high-mobility V2X scenarios \cite{MHC_21}. Furthermore, the complexity of relay selection in group communications can exhibit nondeterministic polynomial (NP)-hard characteristics \cite{np_20}.

Hence, this work is motivated by the premise that the gNB relies on a low-complexity relay selection algorithm for low-latency V2X applications when no SL CSI feedback is available. Furthermore, the relay selection of this work aims to accommodate the mode 1 and mode 2 SL traffic, which to our knowledge, has no precedent in the literature.  That brings the novelty of this work. 
More specifically, in a specified SL channel for the mode 1 SL multicast, the ongoing mode 2 SL transmission may contain a source of interference. The conventional approach to solving this problem is to employ the orthogonal multiple access (OMA) to accommodates mode 1 and mode 2 SL transmission in different time frames. 
On the other hand, the nonorthogonal multiple access (NOMA) can be applied to cooperative network transmission \cite{Sugiura19, Do_noma} for higher spectral efficiency (SE) transmission. That says we may adopt NOMA to accommodate the requirements of mode 1--mode 2 SL transmission for higher SE. The significant contribution in this work is to present an efficient and efficacious approach for determining the employ of the most suitable multiple access (MA) scheme and the associated relay. 

We propose a low-complexity, location-based hybrid MA and relay selection (H-MARS) algorithm for use when no SL CSI information is available. The H-MARS algorithm is executed in two steps. First, the suboptimal relay positions are approximated for NOMA and OMA schemes subject to SL CSI being unavailable. The v-UE closest to the approximated location is specified as the best relay available for each MA scheme. Second, a proposed criterion based on the SE outage of two MA schemes is employed in assisting the selection of MA schemes and their associated relays. H-MARS is suitable for low-latency V2X applications involving no SL CSI feedback. In this algorithm, the relay location approximation and criterion developed for MA selection have low complexity in the order of $\mathcal O(N)$, where $N$ is the number of mode 1 SL multicast members. By comparison, brute-force search (BFS)-based algorithms have a complexity in the order of $\mathcal O(N^2)$. With provided simulation results, we verify that MA selection in the H-MARS algorithm is consistent with those in BFS-based algorithms with significantly lower complexity consumption.

The remainder of this paper is organized as follows. Section \ref{sec:sys_model} lays out the system model and problem formulation. Section \ref{sec:alg} details our proposed algorithm. Section \ref{sec:results} presents the simulation results. Finally, Section \ref{sec:conclusion} concludes the paper.
\section{System Model and Problem Formulation}
\label{sec:sys_model}
 \begin{figure}[!h]
\begin{center}
\includegraphics[clip, trim={5.6cm 1.8cm 2.3cm 3.3cm},width=0.75\linewidth]{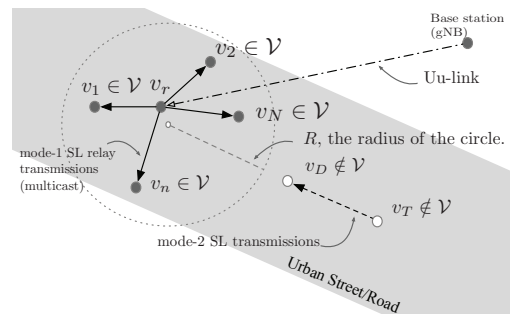}
\vspace*{-2mm}
\caption{System layout.}
\label{fig:sys}
\end{center}
\vspace*{-2mm}
\end{figure}
%
The scenario analyzed in this study is illustrated in Fig. \ref{fig:sys}. Specifically, a BS aims to deliver messages to a designated multicast group $\Vc$ comprising $N$ v-UEs, which are denoted as $\{v_n\in \Vc \}_{n=1}^N$. The v-UE $v_n \in \Vc$ is randomly located within a circle of radius $R$. For  convenience, we take the center of the circle to be the origin in two-dimensional Cartesian coordinates. We denote $\theta_n=(x_n,y_n)$ as the location of $v_n \in \Vc$ in the two-dimensional space. In the present scenario, a v-UE $v_r \in \Vc$ is selected as a relay to relay messages from the BS to the members of a multicast group $\Vc$. In practice, the BS is not within the proximity of multicast group communications within $\Vc$. Thus, per the 3GPP LTE specifications \cite{MHC_21}, the gNB may designate a relay v-UE to execute mode 1 multicast transmission to the designated group members. 
As mentioned, in the specified SL channel, mode 1 multicast transmission may encounter interference from the ongoing mode 2 transmission. As illustrated in Fig. \ref{fig:sys}, we assume  that $v_T \notin \Vc$ is the mode 2--transmitting v-UE closest to the multicast group; moreover, we designate $v_D$ as its corresponding receiving v-UE. The signal transmission from $v_T$ to $v_D$ most likely causes the strongest interference to the ongoing transmission of the multicast group $\Vc$. 
\subsection{OMA Mode}
First, we may schedule the mode 1 multicast SL traffic and mode 2 unicast SL traffic in independent time frames---an approach known as time division MA, a family of OMA schemes. Let $\eta^{\scaleto{OMA}{2.5pt}}_{1}$ denote the SE of the link between a designated relay $v_r\in \Vc$ and a multicast member $v_n\in \Vc;$ it is given by
\begin{align}\label{eq:oma_obj1}
\eta^{\scaleto{OMA}{2.5pt}}_{1,r,n}=\beta {\log_2}(1+g_{r,n}\rho_n d_{r,n}^{-\alpha}),
\end{align}
where $g_{r_n} \sim\exp(1)$ is the exponentially distributed fading power gain of the link between $v_r$ and $v_n$; $d_{r,n}$ is the distance between $v_r$ and $v_n$; $\alpha$ is the path loss exponent; $\rho_n=P/\sigma_n^2$, where $P$ is the transmission power of the v-UE and $\sigma_n^2$ is the noise variance  $v_n$; and $\beta=0.5$ represents  the $50\%$ occupancy time of the specified SL channel. Let $\eta^{\scaleto{OMA}{2.5pt}}_{2}$ denote the SE of the link between $v_T$ and its receiving UE $v_D$, which is given by 
\begin{align}\label{eq:oma_obj2}
\eta^{\scaleto{OMA}{2.5pt}}_{2}=(1-\beta) {\log_2}(1+g_{T,D}\rho_{\scaleto{D}{5pt}} d_{T,D}^{-\alpha}).
\end{align}
In \eqref{eq:oma_obj2}, we use $g_{T,D} \sim\exp(1)$ to denote the exponentially distributed fading power gain of the link between $v_T$ and its receiving UE $v_D$ (spaced $d_{T,D}$ apart) with a rate of 1. We also define $\rho_D$ as $P/\sigma_D^2$, where $P$ is the transmission power of the v-UE and $\sigma_D^2$ is the noise variance $v_D$. 
To ensure the equal distribution of the SL channel between  multicast and unicast SL transmission, the link between $v_T$ and $v_D$ should have a $50\%$ occupancy time of the specified SL channel. Specifically, $(1-\beta)=0.5;   \beta=0.5$. 

To optimize reception under multicast SL communications, we select an optimal relay $v_{\breve r} \in \Vc$ to maximize the total SE of the worst link (i.e., that with the minimum SE) of multicast and unicast transmissions, as indicated in \eqref{eq:oma_obj1} and \eqref{eq:oma_obj2}, respectively. Thus, we obtain the desired SE achieved with the relay $v_{\breve r} \in \Vc,$ given as 
\begin{align}\label{eq:oma_obj}
\eta^{\scaleto{OMA}{2.5pt}}=\max_{\breve r: v_{\breve r}\in\Vc} (\min_{i: v_i\in\Vc} \eta^{\scaleto{OMA}{2.5pt}}_{1,{\breve r},i}+\eta^{\scaleto{OMA}{2.5pt}}_{2}),
\end{align}
where the minimum SE of $\eta^{\scaleto{OMA}{2.5pt}}_{1,{\breve r},i}$ is achieved by $v_i$ and the optimal relay $v_{\breve r} $. 
\subsection{NOMA Mode}
As noted in \cite{Sugiura19, Do_noma}, we can exploit the NOMA to accommodate the multicast--unicast SL transmission  in the specified SL channel. In our scenario, because the multicast members $v_n\in \Vc$ are close to each other, $v_n$ usually receives a stronger signal  from the relay $v_r\in\Vc$ than from $v_T$. Thus, according to the NOMA principle \cite{noma_hanzo}, a v-UE $v_n\in\Vc$ can directly decode the signal from $v_r$. As for $v_D$, its signal from $v_T$ can be decoded using the successive interference cancellation technique.\footnote{According to \cite{TRv2x_r16}, SL CSI feedback is only  supported in unicast transmission. That said, $v_D$ has the CSI of unicast transmission for performing the signal decoding in the NOMA procedure. However, multicast CSI feedback is not supported.}

Under the NOMA scheme \cite{noma_hanzo}, $v_D$ successively decodes and cancels signals transmitted from $v_r\in \Vc$ prior to decoding its own signals from $v_T$. Thus, the SE of the link between the relays $v_r\in \Vc$ and $v_n\in \Vc$ can be written as 
\begin{align}\label{eq:noma_o11}
\eta^{\scaleto{NOMA}{2.5pt}}_{1,r,n}={\log_2}(1+\frac{g_{r,n}\rho_n d_{r,n}^{-\alpha}}{\varrho g_{T,n}\rho_n d_{T,n}^{-\alpha}+1}),
\end{align}
where $\varrho$ is the power factor of NOMA \cite{noma_hanzo} and 
$\rho_n =P/\sigma_n^2$ is defined in \eqref{eq:oma_obj1}.

Let $\eta^{\scaleto{NOMA}{2.5pt}}_{1,r,D}$ denote the SE of the link between the relay $v_r\in \Vc$ and $v_D\notin \Vc$, which is given by
\begin{align}\label{eq:noma_o12}
\eta^{\scaleto{NOMA}{2.5pt}}_{1,r,D}= {\log_2}(1+\frac{g_{r,D}\rho_D d_{r,D}^{-\alpha}}{\varrho g_{T,D}\rho_D d_{T,D}^{-\alpha}+1}),
\end{align}
where $\rho_D=P/\sigma_D^2$ is defined in \eqref{eq:oma_obj2} and we define $\varrho\triangleq\frac{(1+g_{T,D} \rho_D d_{T,D}^{-\alpha})^{0.5}-1}{g_{T,D} \rho_D d_{T,D}^{-\alpha}}$.\footnote{We set $\varrho\triangleq\frac{(1+g_{T,D}\rho_D d_{T,D}^{-\alpha})^{0.5}-1}{g_{T,D}\rho_D d_{T,D}^{-\alpha}}$ assuming $v_D$ can have the same SE in OMA and NOMA, that is, $0.5{\log_2}(1+g_{T,D}\rho_D d_{T,D}^{-\alpha})={\log_2}(1+\varrho g_{T,D}\rho_D d_{T,D}^{-\alpha})$.}
Thus, the SE of the link between the relays $v_T$ and $v_D$ can be written as 
\begin{align}\label{eq:noma_o20}
\eta^{\scaleto{NOMA}{2.5pt}}_{2}={\log_2}(1+\varrho g_{T,D}P d_{T,D}^{-\alpha}).
\end{align}

Similar to the OMA scheme, the optimal relay $v_{\hat r} \in \Vc$ should maximize the total SE of the worst-performing link in multicast and unicast transmission, as described in \eqref{eq:noma_o11} and \eqref{eq:noma_o12}, respectively, to optimize reception in multicast SL communications. Using such an approach, we can obtain the desired SE achieved with the relay $v_{\hat r} \in \Vc,$ which is given as 
\begin{align}\label{eq:noma_obj}
\eta^{\scaleto{NOMA}{2.5pt}}=\max_{{\hat r}:v_{\hat r}\in\Vc} (\min_{i: v_i\in\Vc} (\eta^{\scaleto{NOMA}{2.5pt}}_{1,{\hat r},i}, \eta^{\scaleto{NOMA}{2.5pt}}_{1,{\hat r},D})+\eta^{\scaleto{NOMA}{2.5pt}}_{2}),
\end{align}
where the minimum SE of $\eta^{\scaleto{NOMA}{2.5pt}}_{1,{\breve r},i}$ is achieved by $v_i$ and the optimal relay $v_{\breve r} $. 

On the basis of the results of \eqref{eq:oma_obj} and \eqref{eq:noma_obj}, we can activate the NOMA scheme with the selected relay $v_{\hat r}\in \Vc$ if the SE of NOMA is greater than that of OMA; that is, if $\eta^{\scaleto{NOMA}{2.5pt}} \geq \eta^{\scaleto{OMA}{2.5pt}}$. 

However, the optimal determination of the objective functions of \eqref{eq:oma_obj} and \eqref{eq:noma_obj} requires SL CSI. Moreover, the calculation complexity of conducting a BFS-based approach is at least in the order of $\mathcal O(N^2)$. Compounded with the issue of SL CSI feedback, this is impractical for SL multicast communications and low-latency V2X applications. 
Thus, we introduce a low-complexity location-based approach for situations when SL CSI is unavailable.
\section{Location-Based H-MARS Algorithm}
\label{sec:alg}
\begin{figure}[!h]
\begin{center}
\includegraphics[clip, trim={7cm 4.3cm, 6.2cm 5cm},width=0.7\linewidth]{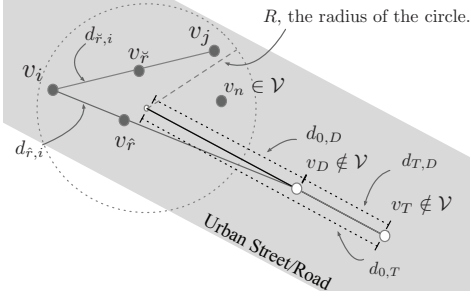}
\vspace*{-2mm}
\caption{Layout of proposed location-based H-MARS algorithm.}
\label{fig:alg}
\end{center}
\vspace*{-2mm}
\end{figure}
The steps of our proposed scheme are detailed as follows. First, according to the locations of all anticipated v-UEs, the suboptimal relay location approximations subject to SL CSI being unavailable are proposed for OMA and NOMA and detailed in Prop. 1 and 2, respectively. 
Second, Prop. 3 addresses the proposed criterion for determining the most suitable MA scheme, derived on the basis of the SE outage of OMA and NOMA, and calculated with the approximated relay locations acquired from Prop. 1 and 2.  Propositions 1--3, the three cornerstones of the proposed H-MARS, are presented as follows.
\begin{prop}
\label{lemma1}
Given the location of $v_D\notin \Vc$, which is spaced $d_{0,D}$ away from the center shown in Fig. \ref{fig:alg}, the suboptimal relay location of OMA subject to SL CSI being unavailable is approximated as 
\begin{align}\label{eq:oma_d}
d_{\breve r,i}= 0.5 d_{i,j},
\end{align}
where $v_i\in\Vc: i=\arg \max_{1\leq n\leq N} d_{n,D}$ is the v-UE with the greatest distance to $v_D$, $v_j\in\Vc$ is the v-UE with the greatest distance to $v_i$, and $d_{i,j}$ is the distance between $v_i$ and $v_j$.
\end{prop}
\begin{IEEEproof}
Because $\eta^{\scaleto{OMA}{2.5pt}}_{2}$ in \eqref{eq:oma_obj} is irrelevant to the selection of relay $v_{\breve r}$, \eqref{eq:oma_obj} is dominated by 
\begin{align}
 \max_{\breve  r: v_{\breve r} \in\Vc} (\min_{i: v_i\in\Vc} \bar\eta^{\scaleto{OMA}{2.5pt}}_{1,{\breve r},i})
 \Rightarrow \min_{{\breve r}: v_{\breve r}\in\Vc} (\max_{i: v_i\in\Vc} d_{{\breve r},i}) \label{eq:oma_obj_d},
\end{align}
where $\bar\eta^{\scaleto{OMA}{2.5pt}}_{1,{\breve r},i}=\beta {\log_2}(1+\rho_i d_{{\breve r},i}^{-\alpha})$ is a modification of \eqref{eq:oma_obj1} given that $g_{{\breve r},i}=1$ because of the lack of SL CSI.
The calculation complexity of \eqref{eq:oma_obj_d} is in the order of $\mathcal O(N^2)$. To find an approximated solution to \eqref{eq:oma_obj_d} for a random topology, we add the following assumptions,
\begin{assump}
$v_i\in\Vc: i=\arg \max_{1\leq n\leq N} d_{n,D}$ is the v-UE with the greatest distance to $v_D$ (Fig. \ref{fig:alg}).
\end{assump}
\begin{assump}
$v_j\in\Vc: j=\arg \max_{1\leq n\leq N}  d_{n,i}$ is the v-UE with the greatest distance to $v_i$ (Fig. \ref{fig:alg}).
\end{assump}
The first assumption gives us an anchor $v_i $. The second assumption allows us to find $v_{\breve r}$ along the line between $v_i$ and $v_j$. Thus, the $v_{\breve r}$ that satisfies \eqref{eq:oma_obj_d} should be in the middle of $v_i$ and $v_j$; that is, $d_{{\breve r},i}= d_{{\breve r},j}= 0.5 d_{i,j}$.
\end{IEEEproof}
\begin{prop}
\label{lemma2}
Given the location of $v_D\notin \Vc$, the suboptimal relay location of NOMA is approximated as 
\begin{align}\label{eq:noma_d}
d_{\hat r,i}=\frac{c_1^{-1/\alpha}}{c_1^{-1/\alpha}+c_2^{-1/\alpha}}\cdot d_{i,D},
\end{align}
where $c_1={\varrho d_{T,i}^{-\alpha}+\rho_i^{-1}}$ and $c_2={\varrho  d_{T,D}^{-\alpha}+\rho_D^{-1}}$.
\end{prop}
\begin{proof}
In \eqref{eq:noma_obj}, $\eta^{\scaleto{NOMA}{2.5pt}}_{2}$ is irrelevant to the selection of relay $v_{\hat r}$; moreover, $\eta^{\scaleto{NOMA}{2.5pt}}$ is dominated by the first term of \eqref{eq:noma_obj}, which is
\begin{align}\label{eq:lemma2_pf_max}
\max_{\hat r:v_{\hat r}\in\Vc} (\min_{i: v_i\in\Vc} (\eta^{\scaleto{NOMA}{2.5pt}}_{1,{\hat r},i}, \eta^{\scaleto{NOMA}{2.5pt}}_{1,{\hat r},D}).
\end{align}
If the inequality $ \eta^{\scaleto{NOMA}{2.5pt}}_{1,{\hat r},i} \geq \eta^{\scaleto{NOMA}{2.5pt}}_{1,{\hat r},D}$ holds, \eqref{eq:lemma2_pf_max} becomes
$\max_{\hat r: v_{\hat r}\in\Vc} \eta^{\scaleto{NOMA}{2.5pt}}_{1,{\hat r},D}$. 
This entails that $v_{\hat r}$ should be as close to $v_D$ as possible, whereas $v_{\hat r}$ should be as far away from $v_i$ as possible, to maximize the minimum of $\eta^{\scaleto{NOMA}{2.5pt}}_{1,\hat r,i}$. Thus, $v_i$ should also be away from $v_D$, making Assumption 1 reasonable.  
Furthermore, $\eta^{\scaleto{NOMA}{2.5pt}}_{1,\hat r,D}$ is upper bounded by the minimum of $\eta^{\scaleto{NOMA}{2.5pt}}_{1,\hat r,i}$. 
To maximize $\eta^{\scaleto{NOMA}{2.5pt}}_{1,\hat r,D}$, we have the associated relay $v_{\hat r}$ to satisfy
\begin{align}
& \eta^{\scaleto{NOMA}{2.5pt}}_{1,{\hat r},i}=\eta^{\scaleto{NOMA}{2.5pt}}_{1,{\hat r},D}\label{eq:prop2_pf}\\ 
\Rightarrow & \frac{g_{\hat r,i}\rho_i d_{\hat r,i}^{-\alpha}}{\varrho g_{T,i}\rho_i d_{T,i}^{-\alpha}+1}=\frac{g_{\hat r,D}\rho_D d_{\hat r,D}^{-\alpha}}{\varrho g_{T,D}\rho_D d_{T,D}^{-\alpha}+1}. \label{eq:noma_pf_equal}
\end{align}
Subject to SL CSI being unavailable, it is reasonable to specify $g_{\hat r,i}, g_{\hat r,D}, g_{T,i}$ and $g_{T,D}$ all equal 1, and we may rewrite \eqref{eq:noma_pf_equal} as
\begin{align}
& \frac{d_{{\hat r},i}^{-\alpha}}{\varrho d_{T,i}^{-\alpha}+\sigma_i^2}=\frac{d_{{\hat r},D}^{-\alpha}}{\varrho  d_{T,D}^{-\alpha}+\sigma_D^2}, \\
\Rightarrow & c_2d_{{\hat r},i}^{-\alpha}=  c_1(d_{i,D}-d_{{\hat r},i})^{-\alpha}  \\
\Rightarrow &  d_{{\hat r},i} = \frac{c_1^{-1/\alpha}d_{i,D}}{c_1^{-1/\alpha}+c_2^{-1/\alpha}},\label{eq:noma_pf_d}
\end{align}
where $c_1={\varrho d_{T,i}^{-\alpha}+\rho_i^{-1}}$ and $c_2={\varrho  d_{T,D}^{-\alpha}+\rho_D^{-1}}$.

Given $v_{\hat r}$, if the inequality of $ \eta^{\scaleto{NOMA}{2.5pt}}_{1,{\hat r},i} < \eta^{\scaleto{NOMA}{2.5pt}}_{1,{\hat r},D}$ holds, \eqref{eq:prop2_pf} can be solved, leading to \eqref{eq:noma_pf_d}. For brevity, this process is not shown.
\end{proof}


%
\begin{figure*}[!t]
\vspace*{-5mm}
     \begin{minipage}[t]{0.32\textwidth}
   \hspace{-3mm}
\includegraphics[width=1.15\linewidth]{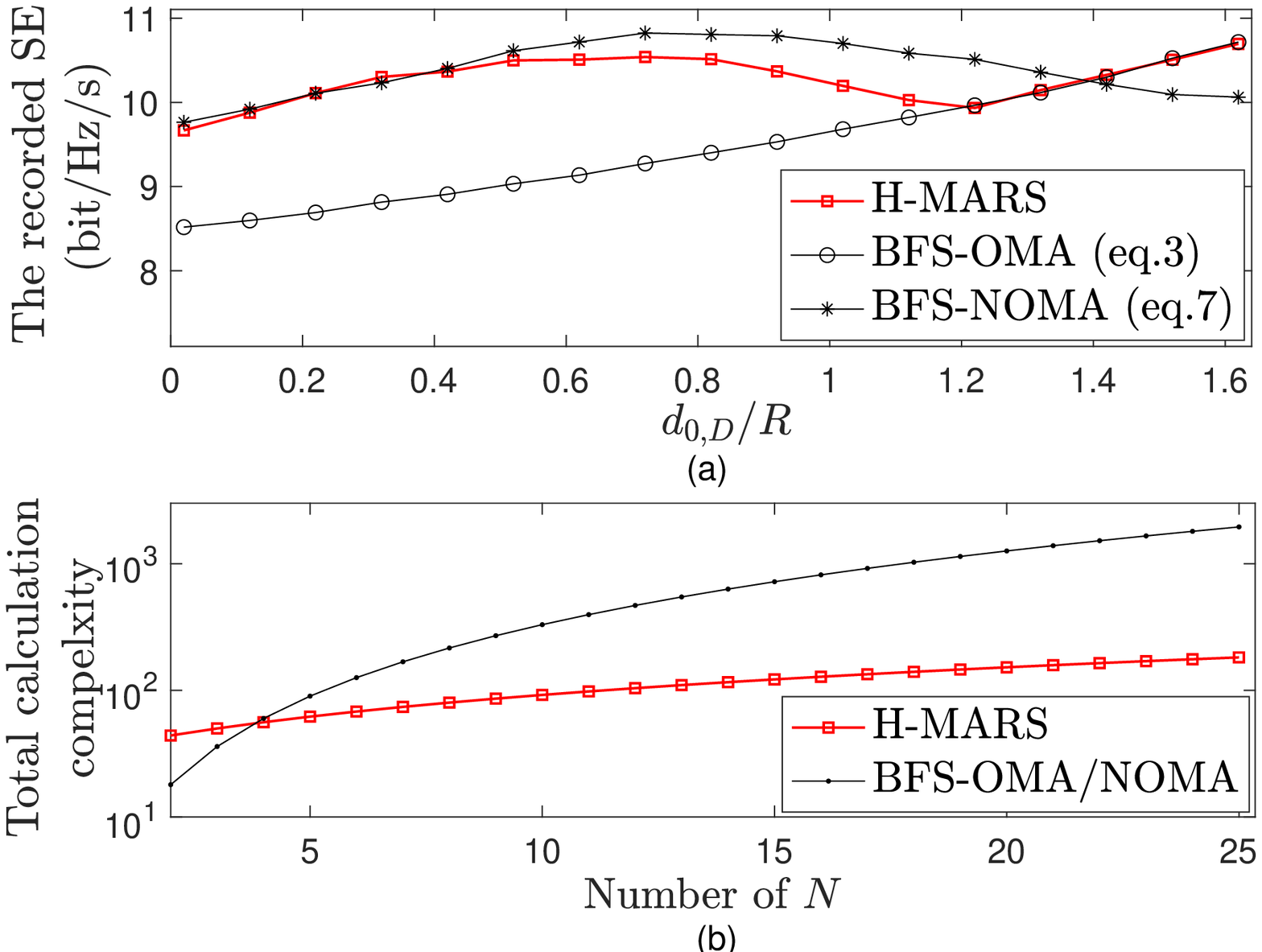}
\vspace*{-5mm}
\caption{\scriptsize{Recorded SEs and total calculation complexity of the H-MARS algorithm with $d_{0,T}/R=3$.}}
\label{fig:1}
\end{minipage}
     \hfill
     \begin{minipage}[t]{0.32\textwidth}
   \hspace{-2mm}
\includegraphics[width=1.15\linewidth]{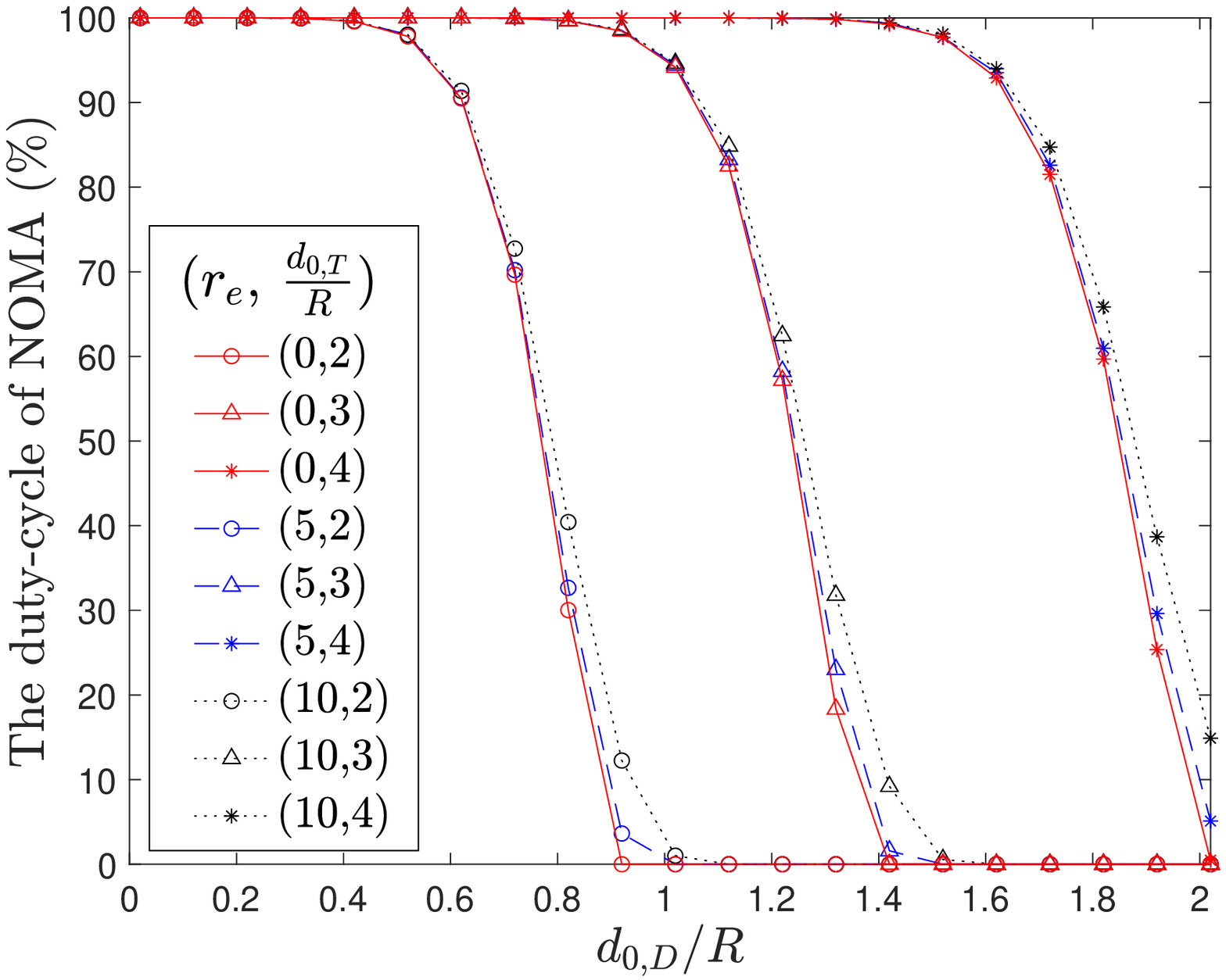}
\vspace*{-5mm}
\caption{\scriptsize{Duty cycle of NOMA in the H-MARS algorithm with different values of $(r_e, d_T/R)$.}}
\label{fig:2}
\end{minipage}
     \hfill
     \begin{minipage}[t]{0.32\textwidth}
\includegraphics[width=1.15\linewidth]{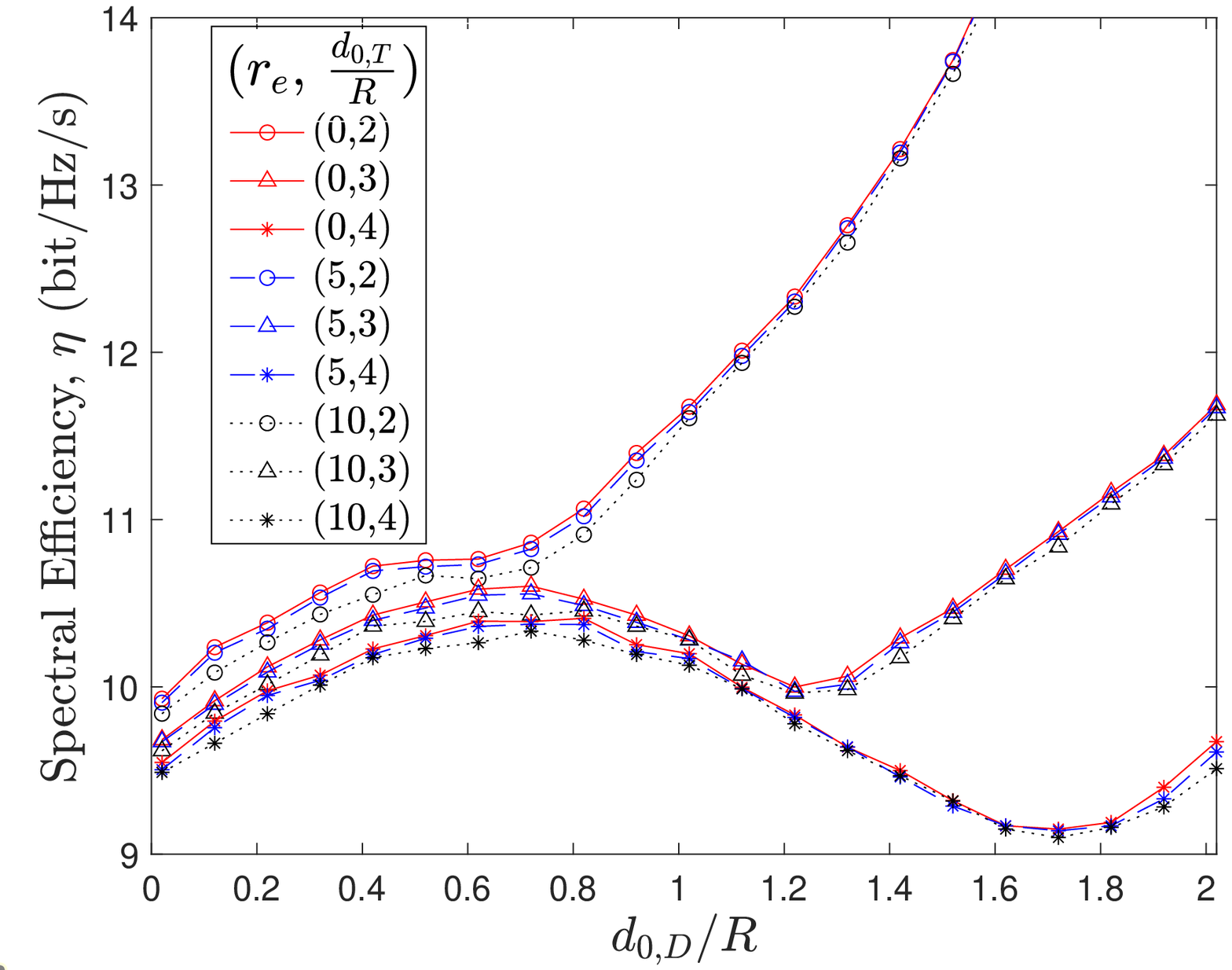}
\vspace*{-5mm}
\caption{\scriptsize{Recorded SEs of the H-MARS algorithm with different values of $(r_e, d_T/R)$.}}
\label{fig:4}
\end{minipage}
\vspace*{-0.3cm}
\end{figure*}
\begin{prop}\label{prop3}
Given the $v_{\breve r}$ and $v_{\hat r}$ selected from Prop. \ref{lemma1} and \ref{lemma2} for the OMA and NOMA schemes, respectively, with a specified $\gamma$, we can activate the NOMA mode if 
\begin{align} \label{eq:prop3} \hspace*{-2mm}
\min (\frac{d_{{\hat r},i}^{-\alpha} e^{-\varrho \rho_i d_{T,i}^{-\alpha} \gamma}}{d_{{\hat r},i}^{-\alpha}+\varrho d_{T,i}^{-\alpha} \gamma}, 
\frac{ d_{{\hat r},D}^{-\alpha} e^{-\varrho \rho_D d_{T,D}^{-\alpha} \gamma}}{ d_{{\hat r},D}^{-\alpha}+\varrho  d_{T,D}^{-\alpha} \gamma}) 
\geq (e^{-\rho_D d_{\breve r,i}^{-\alpha} \gamma})
\end{align} 
is true. Otherwise, the OMA mode will be activated.
\end{prop}
Notably, \eqref{eq:prop3} is dominated by the locations of all anticipated v-UEs and results in the calculation complexity in the order of $\mathcal O(N)$.
\begin{proof}
If the outage probability of NOMA is smaller than that of OMA, the SE of NOMA $\eta^{\scaleto{NOMA}{2.5pt}}$ is likely higher than the SE of OMA $\eta^{\scaleto{OMA}{2.5pt}}$ for a given $\gamma$. In such circumstances, the NOMA mode should be activated; otherwise, the OMA mode should be activated. We have the outage probabilities of NOMA and OMA approximated in Lemma \ref{lemma6} of Appendix \ref{apx_c} and Lemma \ref{lemma7} of Appendix \ref{apx_d}, respectively. Thus, the NOMA mode is activated when 
\begin{align} \hspace*{-2mm}
\max (1-\frac{\rho_i d_{{\hat r},i}^{-\alpha} e^{-\varrho \rho_i d_{T,i}^{-\alpha} \gamma}}{\rho_i d_{{\hat r},i}^{-\alpha}+\varrho \rho_i d_{T,i}^{-\alpha} \gamma},
1-\frac{\rho_D d_{{\hat r},D}^{-\alpha} e^{-\varrho \rho_D d_{T,D}^{-\alpha} \gamma}}{\rho_D d_{{\hat r},D}^{-\alpha}+\varrho \rho_D d_{T,D}^{-\alpha} \gamma}) \nonumber \\
\leq (1-e^{-\rho_D d_{\breve r,i}^{-\alpha} \gamma}),
\end{align} 
which may be rewritten as \eqref{eq:prop3}.
\end{proof}
On the basis of Lemma{\it{}} \ref{lemma4} in Appendix \ref{apx_e},  the association of $\gamma_i$  with $v_i$ is specified to be $\gamma_i= (2^{(A-B)}-1)$, and the association of $\gamma_D$ with $v_D$ can be similarly set. Now, we can specify $\gamma=\min(\gamma_i, \gamma_D)$ in {Prop. 3}.  
The steps of the  H-MARS algorithm are summarized in Algorithm \ref{alg1}.
\setlength{\textfloatsep}{0.1cm}
\setlength{\floatsep}{0.05cm}
\begin{algorithm}[!t]
\footnotesize{
\DontPrintSemicolon 
\KwIn{Locations of $v_i \in \Vc, \forall i$, $v_T$ and $v_D$.}
\KwOut{The determined MA scheme and its associated relay, either $v_{\breve r }\in \Vc$ for NOMA or  $v_{\hat r} \in \Vc$ for OMA. } 
Step 1(a): Find the relay of OMA $v_{\breve r }\in \Vc$, which is the v-UE nearest to the relay location approximation from Prop. 1. \\
Step 1(b): Find the relay of NOMA $v_{\hat r }\in \Vc$, which is the v-UE nearest to the relay location approximation from Prop. 2.  \\
Step 2: Use Prop.{\it{}} 3 to determine whether to activate the NOMA mode with its associated relay $v_{\hat r}$ or the OMA mode with its associated relay $v_{\breve r}$.
}
\caption{Location-Based H-MARS Algorithm}
\label{alg1}
\end{algorithm}
\section{Simulation Results}
\label{sec:results}
The system parameters used in the simulations are summarized as follows. On the basis of the spatial Poisson process \cite{cywei_v2x18}, given a density of 0.25$\%$, $20$ v-UEs on average were randomly distributed in a circle with a radius $R=50(m)$. In addition, we experimented with various locations of $v_T$ and $v_D$. Each recorded SE was the average result of $10^5$  experiments. 
The transmission power $P$ of each v-UE was 21 dBm \cite{TRv2x_r16}, and the noise variance was $-89$ dBm. The path loss exponent $\alpha=4$ is used in Prop. {\it{}}1--3, along with the path loss model presented in \cite{TRv2x_r16}.\footnote{The path loss model used for simulations is $PL=40\log_{10}(d)+7.65-17.3\log_{10}(h_{t}-1)-17.3\log_{10}(h_{r}-1)+2.7\log_{10}(f_c)$, where $h_{t}=h_{r}=1.5$ m and $f_c=5.9$ GHz.}

Fig. \ref{fig:1} displays the recorded SE (a) and the total calculation complexity (b) of the H-MARS and counterpart algorithms in terms of the location of $v_D$, which was determined to exert a strong influence on the results. To our knowledge, there is no other similar work, which also targets a joint design of the MA selection and the associated relay location approximation in V2X networks. Hence, the SEs of the optimal BFS-based OMA and NOMA provide the most suitable performance references to our work. As shown in Fig. \ref{fig:1}(a),  BFS-based NOMA having higher SE when $(d_{0,D}/R)\leq 1.4$ (i.e., $d_{0,D}\leq70m$) is a logical choice over BFS-based OMA. Subject to SL CSI being unavailable, the proposed H-MARS, having the MA switch point at $(d_{0, D}/R)=1.2$, still closely adapts to the changes of BFS-based MA scheme. In addition, the total calculation complexity recorded in Fig. \ref{fig:1}(b) shows that the H-MARS has significantly lower calculation complexity compared to BFS-based MA scheme in terms of the number of multiplications and additions. For example, with the setting of $N=20$, the recorded calculation complexity of H-MARS is at least $20$ times lower than that of the BFS-based MA scheme. These results indicate that the H-MARS algorithm, with low complexity, is extremely efficient and efficacious in guiding the MA scheme selection provided with the effective associated relay location approximation. 

The H-MARS provides an efficacious MA switch. Hence, Fig. \ref{fig:2} presents the duty cycle of NOMA in the H-MARS algorithm in terms of $(r_e, d_{0, T}/R)$, where $r_e$ indicates the maximum value of the uniformly randomized vehicle's location estimation error (in meters), which may be resulted by variant localization algorithms, mobility, etc. \cite{v2xp} As shown in Fig. \ref{fig:2}, the larger value of $d_{0, T}/R$ will result in more duty cycles of NOMA. It entails that OMA is preferred when the mode-2 $v_T$ is closer to the mode-1 V2X group, and NOMA will be employed when $v_T$ is away from the mode-1 V2X group. More importantly, Fig. \ref{fig:2} shows that the MA selection of H-MARS only depends on the relative locations of mode-2 $v_T$ and $v_D$. In addition, Fig. \ref{fig:2} recorded that the location estimation error, especially for $r_e=5$m, results in a negligible impact to the MA's duty-cycle of H-MARS. Hence, the MA schemes do not frequently switch due to rapid CSI variations and have moderate sensitivity to the vehicle's location estimation error. Those facts make the proposed H-MARS practical in use. 

The SEs of H-MARS recorded in Fig. \ref{fig:4} also show the consistency of the above discoveries. The SE of H-MARS presents negligible degradation for $r_e=5$m and shows a moderate loss for $r_e=10$m for all settings of $d_{0,T}/R$. 
More importantly, Fig. \ref{fig:4} shows that with the same value of $d_D$, the mode-2 $v_T$ closer to the mode-1 V2X group, in general, will result in higher SE.
Along with the discovery acquired from Fig. \ref{fig:2}, NOMA enhances the SE with larger $d_{0,T}/R$ and OMA plays its role when $d_{0,T}/R$ is small.
In summary, NOMA is preferred when $v_T$ is far away from the multicast group, whereas OMA performs better when $v_T$ is sufficiently close to that group. The location of $v_D$ in between is key.
\section{Conclusions}
\label{sec:conclusion} 
This paper introduced a low-complexity hybrid MA scheme and its associated relay selection algorithm when SL-CSI is not available.
The proposed scheme has low calculation complexity in the order of $\mathcal O(N)$ compared with the BFS-based algorithm, which is in the order of  $\mathcal O(N^2)$. As a result, the proposed H-MARS algorithm closely adapts to the BFS-based benchmark with low SE loss. Furthermore, it has moderate sensitivity to the vehicle's location estimation error, which makes the proposed H-MARS practical in use.
%
\appendices
%
%
\section{}
\label{apx_c}
\begin{lemma}
\label{lemma6}
Given the selected $v_{\hat r}$ from Prop. \ref{lemma2}, 
we can obtain\begin{align}\label{eq:noma_op}
&\Pr (\eta^{\scaleto{NOMA}{2.5pt}}  < \Gamma)  \nonumber \\
\hspace*{-2mm}
& \propto \max (1-\frac{\rho_i d_{{\hat r},i}^{-\alpha} e^{-\varrho \rho_i d_{T,i}^{-\alpha} \gamma}}{\rho_i d_{{\hat r},i}^{-\alpha}+\varrho \rho_i  d_{T,i}^{-\alpha} \gamma},  
1-\frac{\rho_D d_{{\hat r},D}^{-\alpha} e^{-\varrho \rho_D d_{T,D}^{-\alpha} \gamma}}{\rho_D d_{{\hat r},D}^{-\alpha}+\varrho \rho_D d_{T,D}^{-\alpha} \gamma}),
\end{align}
where $\gamma=(2^\Gamma-1)$.
\end{lemma}
\begin{proof}
According to \eqref{eq:noma_obj}, the selection of $v_{\hat r}$ is irrelevant to $\eta^{\scaleto{NOMA}{2.5pt}}_{2}$ in \eqref{eq:noma_o20}. Thus, we have 
\begin{align}\label{eq:noma_op00}
&\Pr (\eta^{\scaleto{NOMA}{2.5pt}} < \Gamma) \nonumber \\
&\propto \max(\Pr (\eta^{\scaleto{NOMA}{2.5pt}}_{1,{\hat r},i}< \Gamma),\Pr(\eta^{\scaleto{NOMA}{2.5pt}}_{1,{\hat r},D}< \Gamma)).
\end{align}
There exists a $\gamma=(2^\Gamma-1)$ such that
\begin{align}
&\Pr (\eta^{\scaleto{NOMA}{2.5pt}}_{1,{\hat r},i}< \Gamma) \propto  \Pr(\frac{g_{{\hat r},i}\rho_i d_{{\hat r},i}^{-\alpha}}{\varrho g_{T,i}\rho_i d_{T,i}^{-\alpha}+1} < \gamma), \label{eq:noma_op10}\\
&\Pr (\eta^{\scaleto{NOMA}{2.5pt}}_{1,{\hat r},D}< \Gamma) \propto  \Pr(\frac{g_{{\hat r},D}\rho_D d_{r,D}^{-\alpha}}{\varrho g_{T,D}\rho_D d_{T,D}^{-\alpha}+1} < \gamma). \label{eq:noma_op20}
\end{align}

We let $X=g_{{\hat r},i}\rho_i d_{{\hat r},i}^{-\alpha}$ and $Y=\varrho g_{T,i}\rho_i d_{T,i}^{-\alpha}$ with $\lambda_y=\varrho \rho_i d_{T,i}^{-\alpha}$ and $\lambda_x=\rho_i d_{{\hat r},i}^{-\alpha}$ such that 
\begin{align}\label{eq:noma_op11}
\Pr(\frac{g_{{\hat r},i}Pd_{{\hat r},i}^{-\alpha}}{\varrho g_{T,i}\rho_i d_{T,i}^{-\alpha}+1} < \gamma) 
=1-\frac{\rho_i d_{{\hat r},i}^{-\alpha} e^{-\varrho \rho_i d_{T,i}^{-\alpha} \gamma}}{\rho_i d_{{\hat r},i}^{-\alpha}+\varrho \rho_i d_{T,i}^{-\alpha} \gamma}.
\end{align}
Similarly, we have 
\begin{align}\label{eq:noma_op21}
\hspace*{-4mm}
\Pr( \frac{g_{{\hat r},D}\rho_D d_{r,D}^{-\alpha}}{\varrho g_{T,D}\rho_D d_{T,D}^{-\alpha}+1}) < \gamma)
=1-\frac{\rho_D d_{{\hat r},D}^{-\alpha} e^{-\varrho \rho_D d_{T,D}^{-\alpha} \gamma}}{\rho_D d_{{\hat r},D}^{-\alpha}+\varrho \rho_D d_{T,D}^{-\alpha} \gamma}.
\end{align}
In \eqref{eq:noma_op00}, we substitute \eqref{eq:noma_op10} with \eqref{eq:noma_op11} and \eqref{eq:noma_op20} with \eqref{eq:noma_op21}, respectively, to obtain \eqref{eq:noma_op}.
\end{proof}
\section{}
\label{apx_d}
\begin{lemma}
\label{lemma7}
Given the selected $v_{\breve r}$ from Prop. \ref{lemma1}, 
we can obtain\begin{align}\label{eq:oma_op}
\Pr (\eta^{\scaleto{OMA}{2.5pt}}  < \Gamma) \propto 1-e^{-\rho_i d_{{\breve r},i}^{-\alpha} \gamma},
\end{align}
where $\gamma=(2^{\Gamma}-1)$.
\end{lemma}
\begin{proof}
According to \eqref{eq:oma_obj}, the selection of $v_r$ is irrelevant to $\eta^{\scaleto{OMA}{2.5pt}}_{2}$. Thus, we have 
\begin{align}\label{eq:oma_op00}
\Pr (\eta^{\scaleto{OMA}{2.5pt}} < \Gamma)  \propto \Pr (\eta^{\scaleto{OMA}{2.5pt}}_{1,\breve r,i}< \Gamma).
\end{align}
Furthermore, there exists a $\gamma=(2^{\Gamma}-1)$ such that
\begin{align} \label{eq:oma_op10}
\Pr (\eta^{\scaleto{OMA}{2.5pt}}_{1,\breve r,i}< \Gamma) \propto  \Pr(g_{\breve r,i}\rho_i d_{\breve r,i}^{-\alpha} < \gamma).
\end{align}
If $X$ is an exponentially distributed random variable with a rate $\lambda_x$, we have $\Pr(X<x)=1-e^{-\lambda_x x}$. Thus, we let $X=g_{\breve r,i}\rho_i d_{\breve r,i}^{-\alpha}$ with $\lambda_x=\rho_i  d_{\breve r,i}^{-\alpha}$ such that
\begin{align}\label{eq:oma_op11}
\Pr(g_{\breve r,i}\rho_i d_{\breve r,i}^{-\alpha} < \gamma)
=1-e^{-\rho_i d_{\breve r,i}^{-\alpha} \gamma}.
\end{align}
Then, in \eqref{eq:oma_op00}, we substitute \eqref{eq:oma_op10} with \eqref{eq:oma_op11} to obtain \eqref{eq:oma_op}.
\end{proof}
\section{}
\label{apx_e}
\begin{lemma}
\label{lemma4}
Given that $X=g_{{\hat r},i}\rho_i d_{{\hat r},i}^{-\alpha}$ and $Y=\varrho g_{T,i}\rho_i d_{T,i}^{-\alpha}$ are independent and exponentially distributed with rates  $\lambda_y=\varrho \rho_i d_{T,i}^{-\alpha}$ and $\lambda_x=\rho_i d_{{\hat r},i}^{-\alpha}$, respectively, we have
\begin{align} \label{eq:lemma4}
\Eh\{\log_2(1+\frac{X}{Y+1})\}=A-B,
\end{align}
where 
\begin{align}
&A=\frac{e^{1/\lambda_x}E_1(1/\lambda_x)}{\log_e(2)(1-\frac{\lambda_x}{\lambda_y})}+\frac{e^{1/\lambda_y}E_1(1/\lambda_y)}{\log_e(2)(1-\frac{\lambda_y}{\lambda_x})} \text{ and} \\
&B=\frac{e^{1/\lambda_y}E_1(1/\lambda_y)}{\log_e(2)}
\end{align}
\end{lemma}
\begin{proof}
First, \eqref{eq:lemma4} can be rewritten as
\begin{align}\label{eq:lemma4_pf}
\Eh\{\log_2(1+\frac{X}{Y+1})\}=\Eh\{\log_2(1+ (X+Y)\}-\Eh\{\log_2(1+Y)\}.
\end{align}
According to \cite[(C.15)]{bjornson17}, we have
\begin{align}
&\Eh\{\log_2(1+\sum_{k=1}^K g_k )\} =\sum_{k=1}^K \frac{e^{1/\mu_k} E_1(1/\mu_k)}{\log_e(2)\prod_{\ell\neq k}(1-\frac{\mu_\ell}{\mu_k})},
\end{align}
where $g_k$ is exponentially distributed with a rate $\mu_k$ and $E_1(x)=\int_1^\infty e^{-xt}/t dt$ is the exponential integral.

Next, two terms of \eqref{eq:lemma4_pf} are given by 
\begin{align}
&A=\Eh\{\log_2(1+ (X+Y)\}=\frac{e^{1/\lambda_x}E_1(1/\lambda_x)}{\log_e(2)(1-\frac{\lambda_x}{\lambda_y})}+\frac{e^{1/\lambda_y}E_1(1/\lambda_y)}{\log_e(2)(1-\frac{\lambda_y}{\lambda_x})},  \\
&B=\Eh\{\log_2(1+Y)\}=\frac{e^{1/\lambda_y}E_1(1/\lambda_y)}{\log_e(2)}.
\end{align}
\vspace{-2mm}
\end{proof} 
%
\vspace{-2mm}
\bibliographystyle{IEEEtran}
\bibliography{Conf-v1,Book-v1,Journal-v1,v2v}

\begin{thebibliography}{10}
\providecommand{\url}[1]{#1}
\csname url@samestyle\endcsname
\providecommand{\newblock}{\relax}
\providecommand{\bibinfo}[2]{#2}
\providecommand{\BIBentrySTDinterwordspacing}{\spaceskip=0pt\relax}
\providecommand{\BIBentryALTinterwordstretchfactor}{4}
\providecommand{\BIBentryALTinterwordspacing}{\spaceskip=\fontdimen2\font plus
\BIBentryALTinterwordstretchfactor\fontdimen3\font minus
  \fontdimen4\font\relax}
\providecommand{\BIBforeignlanguage}[2]{{%
\expandafter\ifx\csname l@#1\endcsname\relax
\typeout{** WARNING: IEEEtran.bst: No hyphenation pattern has been}%
\typeout{** loaded for the language `#1'. Using the pattern for}%
\typeout{** the default language instead.}%
\else
\language=\csname l@#1\endcsname
\fi
#2}}
\providecommand{\BIBdecl}{\relax}
\BIBdecl

\bibitem{Tdoc_v2x_mc}
RP-202846, ``{NR Sidelink enhancement},'' in \emph{3GPP Rel-17}, Jan 2021.

\bibitem{MHC_21}
M.~H.~C. Garcia, A.~Molina-Galan, M.~Boban, J.~Gozalvez, B.~Coll-Perales,
  T.~Şahin, and A.~Kousaridas, ``{A Tutorial on 5G NR V2X Communications},''
  \emph{IEEE Communications Surveys Tutorials}, vol.~23, no.~3, pp. 1972--2026,
  2021.

\bibitem{cywei_v2x18}
C.-Y. Wei, A.~C.-S. Huang, C.-Y. Chen, and J.-Y. Chen, ``{QoS-Aware Hybrid
  Scheduling for Geographical Zone-Based Resource Allocation in Cellular
  Vehicle-to-Vehicle Communications},'' \emph{IEEE Communications Letters},
  vol.~22, no.~3, pp. 610--613, 2018.

\bibitem{TRv2x_r16}
3GPP, ``{Overall description of Radio Access Network (RAN) aspects for
  Vehicle-to-everything (V2X) based on LTE and NR},'' \emph{{3rd Generation
  Partnership Project (3GPP), TR 37.985}}, Jul. 2020.

\bibitem{cs_mag21}
M.~Harounabadi, D.~M. Soleymani, S.~Bhadauria, M.~Leyh, and E.~Roth-Mandutz,
  ``{V2X in 3GPP Standardization: NR Sidelink in Release-16 and Beyond},''
  \emph{IEEE Communications Standards Magazine}, vol.~5, no.~1, pp. 12--21,
  2021.

\bibitem{np_20}
A.~G. Onalan, E.~D. Salik, and S.~Coleri, ``{Relay Selection, Scheduling, and
  Power Control in Wireless-Powered Cooperative Communication Networks},''
  \emph{IEEE Transactions on Wireless Communications}, vol.~19, no.~11, pp.
  7181--7195, 2020.

\bibitem{Sugiura19}
J.~Kochi, R.~Nakai, and S.~Sugiura, ``{Performance Evaluation of Generalized
  Buffer-State-Based Relay Selection in NOMA-Aided Downlink},'' \emph{IEEE
  Access}, vol.~7, pp. 173\,320--173\,328, 2019.

\bibitem{Do_noma}
D.-T. Do, M.-S.~V. Nguyen, F.~Jameel, R.~Jäntti, and I.~S. Ansari,
  ``{Performance Evaluation of Relay-Aided CR-NOMA for Beyond 5G
  Communications},'' \emph{IEEE Access}, vol.~8, pp. 134\,838--134\,855, 2020.

\bibitem{noma_hanzo}
Y.~Liu, Z.~Qin, M.~Elkashlan, Z.~Ding, A.~Nallanathan, and L.~Hanzo,
  ``{Nonorthogonal Multiple Access for 5G and Beyond},'' \emph{Proceedings of
  the IEEE}, vol. 105, no.~12, pp. 2347--2381, 2017.

\bibitem{v2xp}
H.~Bagheri, M.~Noor-A-Rahim, Z.~Liu, H.~Lee, D.~Pesch, K.~Moessner, and
  P.~Xiao, ``{5G NR-V2X: Toward Connected and Cooperative Autonomous
  Driving},'' \emph{IEEE Communications Standards Magazine}, vol.~5, no.~1, pp.
  48--54, 2021.

\bibitem{bjornson17}
E.~"Bj\"ornson, J.~Hoydis, and L.~Sanguinetti, \emph{Massive MIMO Networks:
  Spectral, Energy, and Hardware Efficiency}.\hskip 1em plus 0.5em minus
  0.4em\relax Now Foundations and Trends, 2017.

\end{thebibliography}
\end{document}